%% file: mainfile.tex
\documentclass[sigconf]{acmart}

\usepackage{booktabs} 
\usepackage{multirow}
\usepackage{graphics}
\usepackage{graphicx}
\usepackage{algorithm}
\usepackage{algorithmic}
\usepackage{amsfonts}
\usepackage{amsmath}
\usepackage{color, url}
\usepackage[skip=0pt]{caption}
\usepackage[skip=0pt]{subcaption}
\usepackage{url}
\usepackage{xcolor, colortbl}
\usepackage{booktabs}
\usepackage{enumitem}
\usepackage{bbm}
\usepackage{dsfont}
\usepackage{amssymb}

\usepackage{bm}

\usepackage{url}

\usepackage{subcaption}

\usepackage[mathscr]{eucal}

\usepackage{bigstrut,multirow,rotating}
\usepackage{placeins}

\graphicspath{ {figs/} }
\usepackage{algorithm,algorithmic}

\usepackage{balance}

\newtheorem{thm}{Theorem}[]

\newcommand{\w}{\bm{w}}
\newcommand{\x}{\bm{x}}
\newcommand{\y}{\bm{y}}
\newcommand{\X}{\bm{X}}
\newcommand{\Y}{\bm{Y}}

\newcommand{\myparatight}[1]{\smallskip\noindent{\bf {#1}:}~}

\copyrightyear{2020}
\acmYear{2020}
\setcopyright{iw3c2w3}
\acmConference[WWW '20]{Proceedings of The Web Conference 2020}{April 20--24, 2020}{Taipei, Taiwan}
\acmBooktitle{Proceedings of The Web Conference 2020 (WWW '20), April 20--24, 2020, Taipei, Taiwan}
\acmPrice{}
\acmDOI{10.1145/3366423.3380072}
\acmISBN{978-1-4503-7023-3/20/04}

\begin{document}

\setlength{\abovedisplayskip}{1.5mm}
\setlength{\belowdisplayskip}{1.5mm}

\title{Influence Function based Data Poisoning Attacks to Top-$N$ Recommender Systems}

\author{Minghong Fang}
\affiliation{%
  \institution{Iowa State University}
}
\email{myfang@iastate.edu}

\author{Neil Zhenqiang Gong}
\affiliation{%
	\institution{Duke University}
}
\email{neil.gong@duke.edu}

\author{Jia Liu}
\affiliation{%
	\institution{Iowa State University}
}
\email{jialiu@iastate.edu}

\begin{abstract}
Recommender system is an essential component of web services to engage users. Popular recommender systems model user preferences and item properties using a large amount of crowdsourced user-item interaction data, e.g., rating scores; then top-$N$ items that match the best with a user's preference are recommended to the user.
In this work, we show that an attacker can launch a \emph{data poisoning attack} to a recommender system to make recommendations as the attacker desires via injecting fake users with carefully crafted user-item interaction data. 
Specifically, an attacker can trick a recommender system to recommend a target item to as many normal users as possible. 
We focus on matrix factorization based recommender systems because they have been widely deployed in industry. 
Given the number of fake users the attacker can inject, we formulate the crafting of rating scores for the fake users as an optimization problem. 
However, this optimization problem is challenging to solve as it is a non-convex integer programming problem. 
To address the challenge, we develop several techniques to approximately solve the optimization problem.  
For instance, we leverage \emph{influence function} to select a subset of normal users who are influential to the recommendations and solve our formulated optimization problem based on these influential users. 
Our results show that our attacks are effective and outperform existing methods.

\end{abstract}



\maketitle

\input{introduction}

\input{related}
\input{problem}

\input{AttackModel}

\input{experiments}

\input{discussion}

\input{conclusion}


\balance
\bibliographystyle{ACM-Reference-Format}
\bibliography{refs,refs1,refs2}

\end{document}

%% file: introduction.tex

\section{Introduction} \label{sec:intro}

Recommender system is a key component of many web services to help users locate items they are interested in. Many recommender systems are based on collaborative filtering. For instance, given a large amount of user-item interaction data (we consider rating scores in this work) provided by users, a recommender system learns to model latent users' preferences and items'  features, and then the system recommends top-$N$ items to each user, where the features of the top-$N$ items best match with the user's preference. 

As a recommender system is driven by user-item interaction data, an attacker can manipulate a recommender system via injecting fake users with fake user-item interaction data to the system. Such attacks are known as \emph{data poisoning attacks}~\cite{lam2004shilling,mobasher2007toward,wilson2013power,li2016data,yang2017fake,fang2018poisoning,fang2019local}. 
Several recent studies designed recommender-system-specific data poisoning attacks to association-rule-based~\cite{yang2017fake}, graph-based ~\cite{fang2018poisoning} and matrix-factorization-based recommender systems~\cite{li2016data}. 
However, how to design customized attacks to matrix-factorization-based top-$N$ recommender systems remains an open question even though such recommender systems have been widely deployed in the industry. 
In this work, we aim to bridge the gap. In particular, we aim to design an optimized data poisoning attack to {\em matrix-factorization-based top-$N$ recommender systems}. 
Suppose that an attacker can inject $m$ fake users into the recommender system and each fake user can rate at most $n$ items, which we call \emph{filler items}. Then, the key question is: {\em how to select the filler items and assign  rating scores to them such that an attacker-chosen target item is recommended to as many normal users as possible?} 
To answer this question, we formulate an optimization problem for selecting filler items and assigning rating scores for the fake users, with an objective to maximize the number of normal users to whom the target item is recommended.

However, it is challenging to solve this optimization problem because it is a non-convex integer programming problem. 
To address the challenge, we propose a series of techniques to approximately solve the optimization problem. 
First, we propose to use a loss function to approximate the number of normal users to whom the target item is recommended. 
We relax the integer rating scores to continuous variables and convert them back to integer rating scores after solving the reformulated optimization problem.  
Second, to enhance the effectiveness of our attack, we leverage the {\em influence function approach} inspired by the interpretable machine learning literature \cite{koh2017understanding,wang2018data,koh2019accuracy} to account for the reality that the top-$N$ recommendations may be only affected by a subset $\mathcal{S}$ of influential users.
For convenience, throughout the rest of this paper, we refer to our attack as $\mathcal{S}$-attack.
We show that the influential user selection subproblem enjoys the {\em submodular} property, which guarantees a $(1-1/e)$ approximation ratio with a simple greedy selection algorithm.
Lastly, given $\mathcal{S}$, we develop a gradient-based optimization algorithm to determine rating scores for the fake users.

We evaluate our $\mathcal{S}$-attack and compare it with multiple baseline attacks on two benchmark datasets, including Yelp and Amazon Digital Music (Music). 
Our results show that our attacks can effectively promote a target item. For instance, on the Yelp dataset, when injecting only 0.5\% of fake users, our attack can make a randomly selected target item appear in the top-$N$ recommendation lists of 150 times more normal users. 
Our $\mathcal{S}$-attack outperforms the baseline attacks and continues to be effective even if the attacker does not know the parameters of the target recommender system.
We also investigate the effects of our attacks on recommender systems that are equipped with fake users detection capabilities.
For this purpose, we train a binary classifier to distinguish between fake users and normal ones. 
Our results show that this classifier is effective against traditional attack schemes, e.g., PGA attack \cite{li2016data}, etc.
Remarkably, we find that our influence-function-based attack continues to be effective.
The reason is that our proposed attack is designed with stealth in mind, and
the detection method can detect some fake users but miss a large fraction of them. 

Finally, we show that our influence function based approach can also be used to enhance data poisoning attacks to graph-based top-$N$ recommender systems. Moreover, we show that instead of using influence function to select a subset of influential users, using influence function to weight each normal user can further improve the effectiveness of data poisoning attacks, though such approach sacrifices computational efficiency.

In summary, our contributions are as follows:
\begin{list}{\labelitemi}{\leftmargin=1em}
\item We propose the first data poisoning attack to matrix-factorization-based Top-$N$ recommender systems, which we formulate as a non-convex integer optimization problem.   

\item We propose a series of techniques to approximately solve the optimization problem with provable performance guarantee. 

\item We evaluate our $\mathcal{S}$-attack and compare it with state-of-the-art using two benchmark datasets. Our results show that our attack is effective and outperforms existing ones.

\end{list}

%% file: related.tex

\vspace{-.1in}
\section{Related Work} \label{sec:related}

\myparatight{Data poisoning attacks to recommender systems} The security and privacy issues in machine learning models have been studied in many scenarios ~\cite{yin2018byzantine,shafahi2018poison,shokri2017membership,nasr2018comprehensive,wang2018stealing,yang2019byzantine,zhang2020private}.
The importance of data poisoning attacks has also been recognized in recommender systems \cite{christakopoulou2019adversarial,xing2013take,mobasher2005effective,seminario2014attacking,mehta2008attack,mobasher2007toward}.
Earlier work on poisoning attacks against recommender systems are mostly agnostic to recommender systems and do not achieve satisfactory attack performance, e.g., random attack \cite{lam2004shilling} and average attack \cite{lam2004shilling}. 
Recently, there is a line of work focusing on attacking specific types of recommender systems \cite{yang2017fake,fang2018poisoning, li2016data}. 
For example, 
Fang {\em et al.}~\cite{fang2018poisoning} proposed efficient poisoning attacks to graph-based recommender systems. They injected fake users with carefully crafted rating scores to the recommender systems in order to promote a target item. They modeled the attack as an optimization problem to decide the rating scores for the fake users. 
Li {\em et al.}~\cite{li2016data} proposed poisoning attacks to matrix-factorization-based recommender systems. Instead of attacking the top-$N$ recommendation lists, their goal was to manipulate the predictions for all missing entries of the rating matrix.   
As a result, the effectiveness of their attacks is unsatisfactory in matrix-factorization-based top-$N$ recommender systems.

\myparatight{Data poisoning attacks to other systems} Data poisoning attacks generally refer to attacks that manipulate the training data of a machine learning or data mining system such that the learnt model makes predictions as an attacker desires. Other than recommender systems, data poisoning attacks were also studied for other systems. For instance, existing studies have demonstrated effective data poisoning attacks can be launched to anomaly detectors~\cite{rubinstein2009antidote}, spam filters~\cite{Nelson08poisoningattackSpamfilter}, SVMs~\cite{biggio2012poisoning,xiao2012adversarial},    regression methods~\cite{xiao2015feature,Jagielski18}, graph-based methods~\cite{zugner2018adversarial,Wang19}, neural networks~\cite{Gu17,Liu18,Chen17}, and federated learning~\cite{fang2019local}, which significantly affect their performance.

%% file: problem.tex

\section{Problem Formulation} \label{sec:problem}

\subsection{Matrix-Factorization-Based Recommender Systems: A Primer}
A matrix-factorization-based recommender system \cite{koren2009matrix} maps users and items into latent factor vectors. 
Let $\mathcal{U}$, $\mathcal{I}$ and $\mathcal{E}$ denote the user, item and rating sets, respectively. 
We also let $\vert \mathcal{U} \vert$, $\vert \mathcal{I} \vert$ and $\vert \mathcal{E} \vert$ denote the numbers of users, items and ratings, respectively. 
Let $ \bm{R} \in {{\mathbb{R}}^{\vert \mathcal{U} \vert \times \vert \mathcal{I} \vert}}$ represent the user-item rating matrix, where each entry ${r_{ui}}$ denotes the score that user $u$ rates the item $i$. 
Let $\bm{x}_u \in {\mathbb{R}^d}$ and $\bm{y}_i \in {\mathbb{R}^d}$ denote the latent factor vector for user $u$ and   item $i$, respectively, where $d$ is the dimension of latent factor vector.
For convenience, we use matrices $\X = [\x_{1},\ldots,\x_{|\mathcal{U}|}]$ and $\Y = [\y_{1},\ldots,\y_{|\mathcal{I}|}]$ to group all $\x$- and $\y$-vectors.
In matrix-factorization-based recommender systems, we aim to learn $\X$ and $\Y$ via solving the following optimization problem:
\begin{align}
\label{orig_matrix_opti_problem}
\mathop {\arg \min }\limits_{{\bm{X}},{\bm{Y}}} \!\!\! {\sum\limits_{(u,i) \in \mathcal{E}} {\left( {{r_{ui}} - \boldsymbol{x}_u^\top{\bm{y}_i}} \right)} ^2} + \lambda \left( {\sum\limits_u {{\lVert \bm{x}_u \rVert}_2^2 } } +  {\sum\limits_i {{\lVert \bm{y}_i \rVert}_2^2 } } \right),
\end{align}
where ${\lVert \cdot \rVert}_2$ is the $\ell_2$ norm and $\lambda$ is the regularization parameter. 
Then,  the rating score that a user $u$ gives to an unseen item $i$ is predicted  as $\hat{r}_{ui} = \bm{x}_u^\top{\bm{y}_i}$, where $\bm{x}_u^\top$ denotes the transpose of vector $\bm{x}_u$.
Lastly, the $N$ unseen items with the highest predicted rating scores are recommended to each user.

\subsection{Threat Model}
Given a target item $t$, the goal of the attacker is to promote item $t$ to as many normal users as possible and maximize the hit ratio $h(t)$, which is defined as the fraction of normal users whose top-$N$ recommendation lists include the target item $t$. 
We assume that the attacker is able to inject some fake users into the recommender system, each fake user will rate the target item $t$ with high rating score and give carefully crafted rating scores to other well-selected items.
The attacker may have full knowledge of the target recommender system (e.g., all the rating data, the recommendation algorithm). The attacker may also only have partial knowledge of  the target recommender system, e.g., the attacker only has access to some ratings. We will  show that our attacks are still effective when the attacker has partial knowledge of the target recommender system.

\subsection{Attack Strategy}

We assume that the rating scores of the target recommender system are integer-valued and can only be selected from the set $\{0, 1, \cdots,$  $r_{max}\}$, where $r_{max}$ is the maximum rating score. 
We assume that the attacker can inject $m$ fake users into the recommender system. We denote by $\mathcal{M}$ the set of $m$ fake users. Each fake user will rate the target item $t$ and at most $n$ other carefully selected items (called \emph{filler items}). We consider each fake user rates at most $n$ filler items to avoid being easily detected. 
We let $\bm{r}_{v}$ and $\Omega _v$ denote the rating score vector of fake user $v$ and the set of items rated by $v$, respectively, where $v \in \mathcal{M}$ and $\vert {\Omega _v} \vert \le n+1$.  
Then, $r_{vi}$ is the score that user $v$ rates the item $i$, $i \in \Omega _v$. 
Clearly, ${\Omega _v}$ satisfies $\vert {\Omega _v} \vert = \lVert {\bm{r}_v} \rVert _0$, where ${\lVert \cdot \rVert}_0$ is the $\ell_0$ norm (i.e., the number of non-zero entries in a vector).
The attacker's goal is to find an optimal rating score vector for each fake user $v$ to maximize the hit ratio $h(t)$.
We formulate this hit ratio maximization problem (HRM) as follows:
\begin{align}
\label{objectiveFrame}
\text{HRM:  max}  & \medspace\medspace \textstyle h(t)  &  \\
\label{budgetconstraint}
\text{s.t. } \hspace{-0.25em} & \medspace\medspace \vert {\Omega _v} \vert  \le n+1,    && \forall v \in \mathcal{M}, \hspace{0.25em} &  \\
\label{itemSelected}
& \medspace \medspace r_{vi} \in \{0,1,\cdots,r_{max}\},   && \forall v \in \mathcal{M}, \forall i \in \Omega _v.
\end{align}
Problem HRM is an integer programming problem and is NP-hard in general.
Thus, finding an optimal solution is challenging.
In the next section, we will propose techniques to approximately solve the problem.

%% file: AttackModel.tex

\section{Our Solution} \label{AttackModel}

We optimize the rating scores for fake users one by one instead of optimizing for all the $m$ fake users simultaneously. In particular, we repeatedly optimize the rating scores of one fake user and add the fake user to the recommender system until we have $m$ fake users. However, it is still challenging to solve the HRM problem even if we consider only one fake user. To address the challenge, we design several techniques to approximately solve the HRM problem for one fake user. First, we relax the discrete ratings to continuous data and convert them back to discrete ratings after solving the problem.  Second, we use a differentiable loss function to approximate the hit ratio. Third, instead of using all normal users, we use a selected subset of influential users to solve the HRM problem, which makes our attack more effective. Fourth, we develop a gradient-based method to solve the HRM problem to determine the rating scores for the fake user.  

\subsection{Relaxing Rating Scores} 
We let vector $\bm{w}_{v} = [w_{vi}, i \in \Omega_v]^{\top}$ be the relaxed continuous rating score vector of fake user $v$, where $w_{vi}$ is the rating score that user $v$ gives to the item $i$. 
Since $r_{vi} \in \{0,1,\cdots,r_{max}\}$ is discrete, which makes it difficult to solve the optimization problem defined in (\ref{objectiveFrame}), we relax the discrete rating score $r_{vi}$ to continuous variables $w_{vi}$ that satisfy $w_{vi} \in [0,r_{max}]$.
Then, we can use  gradient-based methods to compute $\bm{w}_{v}$. 
After we solve the optimization problem, we convert each $w_{vi}$ back to a discrete integer value in the set $ \{0,1,\cdots,r_{max}\}$.

\subsection{Approximating the Hit Ratio}

We let $\Gamma_u$ be the set of top-$N$ recommended items for a user $u$, i.e., $\Gamma_u$ consists of the $N$ items that $u$ has not rated before and have the largest predicted rating scores. 
To approximate the optimization problem defined in (\ref{objectiveFrame}),
 we define a  loss function that is subject to the following rules: 1) for each item $i \in \Gamma_u$, if $\hat{r}_{ui} < \hat{r}_{ut}$, then the loss is small, where $\hat{r}_{ui}$ and $\hat{r}_{ut}$ are the predicted rating scores that user $u$ gives to item $i$ and target item $t$, respectively; 
2) the higher target item $t$ ranks in $\Gamma_u$, the smaller the loss. 
Based on these rules, we reformulate the HRM problem as the following problem:
\begin{align}
\label{appproblem_orginal}
\begin{split}
& \min_{\bm{w}_v} \mathcal{L}_{\mathcal{U}}(\bm{w}_v) = \sum_{u \in \mathcal{U}} {\sum_{i \in \Gamma_u}  g(\hat{r}_{ui} - \hat{r}_{ut})} +  \eta {\left\| \bm{w}_v \right\|_1} \\
& \text{   s.t. } w_{vi} \in [0, r_{max}],
\end{split}
\end{align}
where $g(x) = \frac{1}{{1 + \exp (-{x/b})}}$ is the Wilcoxon-Mann-Whitney loss function~\cite{backstrom2011supervised}, $b$ is the \emph{width} parameter, $\eta$ is the regularization parameter, and ${\lVert \cdot \rVert}_1$ is the $\ell_1$ norm. 
Note that $g(\cdot)$ guarantees that $\mathcal{L}_{\mathcal{U}}(\bm{w}_v) \geq 0$ and is differentiable. 
The $\ell_1$ regularizer $\left\| \bm{w}_v \right\|_1$ aims to model the constraint that each fake user rates at most $n$ filler items. In particular, 
the $\ell_1$ regularizer makes a fake user's ratings small to many items and we can select the $n$ items with the largest ratings as the filler items.

\subsection{Determining the Set of Influential Users}

It has been observed in \cite{lapedriza2013all, wang2018data} that different training samples have different  contributions to the solution quality of an optimization problem, and the performance of the model training could be improved if we drop some training samples with low contributions. 
Motivated by this observation,  instead of optimizing the ratings of a fake user over all normal users, we solve the problem in (\ref{appproblem_orginal}) using a subset of \emph{influential users}, who are the most responsible for the prediction of  the target item before attack. 
We let $\mathcal{S} \in \mathcal{U}$ represent the set of influential users for the target item $t$.
For convenience, in what follows, we refer to our attack as $\mathcal{S}$-attack.
Under the $\mathcal{S}$-attack, we further reformulate (\ref{appproblem_orginal}) as the following problem:
\begin{align}
\label{appproblem}
\begin{split}
& \min_{\bm{w}_v} \mathcal{L}_{\mathcal{S}}(\bm{w}_v) = \sum_{u \in \mathcal{S}} {\sum_{i \in \Gamma_u}  g(\hat{r}_{ui} - \hat{r}_{ut})} +  \eta {\left\| \bm{w}_v \right\|_1} \\
& \text{   s.t. } w_{vi} \in [0, r_{max}].
\end{split}
\end{align}

Next, we propose an influence function approach to determine $\mathcal{S}$
and then solve the optimization problem defined in (\ref{appproblem}). 
We let $\digamma(\mathcal{S},t)$ denote the influence of removing all users in the set $\mathcal{S}$ on the prediction at the target item $t$, where influence here is defined as the change of the predicted rating score. We want to find a set of influential users that have the largest influence on the target item $t$. Formally, the influence maximization problem can be defined as:
\begin{align}
\label{influ_opt}
\text{max}   \medspace\medspace \textstyle \digamma(\mathcal{S},t), \quad 
\text{s.t. } \hspace{-0.5em}  \medspace\medspace \vert {\mathcal{S}} \vert  = \Delta,
\end{align}
where $\Delta$ is the desired set size (i.e., the number of users in set $\mathcal{S}$). 
However, it can be shown that the problem is NP-hard \cite{kempe2003maximizing}.
In order to solve the above influence maximization problem of (\ref{influ_opt}), we first show how to measure the influence of one user, then we show how to approximately find a set of $\Delta$ users with the maximum influence.

We define $\pi(k,t)$ as the influence of removing user $k$ on the prediction at the target item $t$:
\begin{align}
\label{influence_user_target}
\pi(k,t) & \stackrel{\text{def}} = \sum\nolimits_{j \in {\Omega _k}} {\varphi((k,j),t)},
\end{align}
where $\varphi((k,j),t)$ is the influence of removing edge $(k,j)$ in the user-item bipartite on the prediction at the target item $t$, ${\Omega _k}$ is the set of items rated by user $k$.
Then, the influence of removing user set $\mathcal{S}$ on the prediction at the target item $t$ can be defined as:
\begin{align}
\label{influ_set_S}
\digamma(\mathcal{S},t) & \stackrel{\text{def}} = \sum\nolimits_{k \in \mathcal{S}} {\pi(k,t)}.
\end{align}

Since the influence of user and user set can be computed based on the edge influence $\varphi((k,j),t)$, the key challenge boils down to how to evaluate $\varphi((k,j),t)$ efficiently. 
Next, we will propose an appropriate influence function to efficiently compute $\varphi((k,j),t)$.

\subsubsection{Influence Function for Matrix-factorization-based Recommender Systems}
For a given matrix-factorization-based recommender system, we can rewrite (\ref{orig_matrix_opti_problem}) as follows:
\begin{align}
\label{matrix_square_loss_influ}
{\bm{\theta} ^*} = \mathop {\arg \min }\limits_{\bm{\theta} } \frac{1}{{|\mathcal{E}|}} \sum\limits_{(u,i) \in \mathcal{E}} {\ell((u,i),\bm{\theta} )},
\end{align}
where $\bm{\theta} \!\triangleq\! (\bm{X},\bm{Y})$.
We let ${\hat{r}_{ui}}(\bm{\theta} )$ denote the predicted rating score  user $u$ gives to item $i$ under parameter $\bm{\theta}$, and ${\hat{r}_{ui}}(\bm{\theta} ) \triangleq  {\bm{x}_u^\top}(\bm{\theta} ){\bm{y}_i}(\bm{\theta} )$.

If we increase the weight of the edge $(k,j) \in \mathcal{E} $ by some $\zeta$, then the perturbed optimal parameter $\bm{\theta} _{\zeta ,(k,j)}^*$ can be written as:
\begin{align}
\bm{\theta} _{\zeta ,(k,j)}^* = \mathop {\arg \min }\limits_{\bm{\theta}}  \frac{1}{{|\mathcal{E}|}} \sum\limits_{(u,i) \in \mathcal{E}} {\ell((u,i),\bm{\theta} )} + \zeta \ell((k,j),\bm{\theta} ).
\end{align}

Since removing the edge $(k,j)$ is equivalent to increasing its weight by $\zeta =-\frac{1}{{|\mathcal{E}|}}$, the influence of removing edge $(k,j)$ on the prediction at edge $(o,t)$ can be approximated as follows \cite{cook1980characterizations,koh2017understanding}:
\begin{align}
\Phi((k,j),(o,t)) 
& \!\stackrel{\text{def}}= \!{\hat{r}_{ot}} \! \left( {\bm{\theta} _{\varepsilon  \setminus (k,j)}^*} \right) - {\hat{r}_{ot}}({\bm{\theta} ^*}) 
\! \approx \! -\frac{1}{{|\mathcal{E}|}} \! \cdot \!\! {\left. {\frac{{\partial {\hat{r}_{ot}}\left(\bm{\theta} _{\zeta ,(k,j)}^*\right)}}{{\partial \zeta }}} \right|_{\zeta  = 0  }}  \nonumber  \\
&= \frac{1}{{|\mathcal{E}|}} {\nabla _{\bm{\theta}} }{\hat{r}_{ot}^\top}({\bm{\theta} ^*}) \bm{H}_{{\bm{\theta} ^*}}^{ - 1} {\nabla _{\bm{\theta}} }\ell((k,j),{\bm{\theta} ^*}), 
\end{align}
where ${\bm{\theta} _{\varepsilon  \setminus (k,j)}^*}$ is the optimal model parameter after removing edge $(k,j)$ and ${\bm{H}_{{\bm{\theta} ^*}}}$ represents the Hessian matrix of the objective function defined in (\ref{matrix_square_loss_influ}). 
Therefore, the influence of removing edge $(k,j)$ on the prediction at the target item $t$ can be computed as:
\begin{align}
\varphi((k,j),t)  & \stackrel{\text{def}}= \sum\nolimits_{o \in \mathcal{U}} \left| {\Phi((k,j),(o,t))} \right|, 
\end{align}
where $\left| \cdot \right|$ is the absolute value.

\subsubsection{Approximation Algorithm for Determining $\mathcal{S}$} \label{sec:Approximating_influential_user_set}

Due to the combinatorial complexity, solving the optimization problem defined in (\ref{influ_opt}) remains an NP-hard problem. 
Fortunately, based on the observation that the influence of set $\mathcal{S}$ (e.g., $\digamma(\mathcal{S},t)$) exhibits a diminishing returns property, we propose a greedy selection algorithm to find a solution to (\ref{influ_opt}) with an approximation ratio guarantee. 
The approximation algorithm is a direct consequence of the following result, which says that the influence $\digamma(\mathcal{S},t)$ is monotone and submodular.
\begin{thm}[Submodularity] \label{monotonically_non_decreasing}
	The influence $\digamma(\mathcal{S},t)$ is normalized, monotonically non-decreasing and submodular.
\end{thm}
\begin{proof}
	Define three sets $\mathcal{A}$, $\mathcal{B}$ and $\mathcal{C}$, where $\mathcal{A} \subseteq \mathcal{B}$ and $\mathcal{C} = \mathcal{B} \setminus \mathcal{A}$. 
	To simplify the notation, we use $\digamma(\mathcal{A})$ to denote $\digamma(\mathcal{A},t)$. 
	It is clear that the influence function is normalized since $\digamma(\emptyset) = 0$.
	When there is no ambiguity, we let $\digamma(u)$ denote $\digamma(\{ u \})$ for $u \in \mathcal{U}$.
	Since
	$
	\digamma(\mathcal{B}) - \digamma(\mathcal{A}) 
	= \sum\limits_{u \in \mathcal{B}} {\digamma(u)} - \!\! \sum\limits_{u \in \mathcal{A}} {\digamma(u)} 
	= \!\!\!\! \sum\limits_{u \in \mathcal{B} \setminus \mathcal{A} } \!\! {\digamma(u)}
	= \digamma(\mathcal{C})
	\ge 0 \nonumber,
	$
	which implies that the influence $\digamma(\mathcal{S},t)$ is monotonically non-decreasing.
	To show the submodular property, we let $\overline{\mathcal{S}}$ denote the complement of a set $\mathcal{S}$.
	Now, consider an arbitrary set $\mathcal{D}$, for which we have:
	$
	\digamma(\mathcal{B} \cup \mathcal{D}) - \digamma(\mathcal{A} \cup \mathcal{D})
	= \digamma((\mathcal{B} \cup \mathcal{D}) \setminus (\mathcal{A} \cup \mathcal{D})) 
	\stackrel{(a)} = \digamma(\mathcal{C} \setminus (\mathcal{C} \cap \mathcal{D})) \le \digamma(\mathcal{C}) = \digamma(\mathcal{B}) - \digamma(\mathcal{A}), \nonumber
	$
	where $(a)$ follows from $(\mathcal{B} \cup \mathcal{D}) \setminus (\mathcal{A} \cup \mathcal{D}) = (\mathcal{B} \cup \mathcal{D}) \cap \overline{(\mathcal{A} \cup \mathcal{D})} = (\mathcal{B} \cup \mathcal{D}) \cap (\overline{\mathcal{A}} \cap \overline{\mathcal{D}}) = (\mathcal{B} \cap \overline{\mathcal{A}} \cap \overline{\mathcal{D}}) \cup (\mathcal{D} \cap \overline{\mathcal{A}} \cap \overline{\mathcal{D}}) = \mathcal{C} \cap \overline{\mathcal{D}} = \mathcal{C} \setminus (\mathcal{C} \cap \mathcal{D})$.
	Hence, the influence $\digamma(\mathcal{S},t)$ is submodular and the proof is completed.
\end{proof}	

Based on the submodular property of $\digamma(\mathcal{S},t)$, we propose Algorithm \ref{Find_Influential_user_Set}, a greedy-based selection method to select an influential user set $\mathcal{S}$ with $\Delta$ users.
More specifically, we first compute the influence of each user, and add the user with the largest influence to the candidate set $\mathcal{S}$ (breaking ties randomly). 
Then, we recompute the influence of the remaining users in the set $\mathcal{U} \setminus \mathcal{S}$, and find the user with the largest influence within the remaining users, so on and so forth.
We repeat this process until we find $\Delta$ users.
Clearly, the running time of Algorithm~\ref{Find_Influential_user_Set}  is linear.
The following result states that Algorithm~\ref{Find_Influential_user_Set} achieves a $(1-1/e)$ approximation ratio, and its proof follows immediately from standard results in submodular optimization \cite{nemhauser1978analysis}
 and is omitted here for brevity.
\begin{thm} \label{greedy_approximate}
	Let $\mathcal{S}$ be the influential user set returned by Algorithm~\ref{Find_Influential_user_Set} and let $\mathcal{S^*}$ be the optimal influential user set, respectively. It then holds that
	$\digamma(\mathcal{S},t) \ge \left( {1 - \frac{1}{e}} \right) \digamma(\mathcal{S^*},t) \nonumber$.
\end{thm}

\begin{algorithm}[!t]
	\caption{\textit{Greedy Influential User Selection.}}\label{Find_Influential_user_Set}
	\begin{algorithmic}[1]
		\renewcommand{\algorithmicrequire}{\textbf{Input:}}
		\renewcommand{\algorithmicensure}{\textbf{Output:}}
		\REQUIRE  Rating matrix $\bm{R}$, budget $\Delta$.
		\ENSURE  Influential user set $\mathcal{S}$.
		\STATE Initialize $\mathcal{S} = \emptyset $.
		\WHILE {$\mathcal{|S|} < \Delta$}
		\STATE Select $u = \arg {\max _{k \in \mathcal{U} \setminus \mathcal{S}}} \pi(k,t)$.
		\STATE $\mathcal{S} \leftarrow \mathcal{S} \cup \{u\}$.
		\ENDWHILE
		\RETURN $\mathcal{S}$. 
	\end{algorithmic} 
\end{algorithm}

\vspace{-.1in}
\subsection{Solving Rating Scores for a Fake User}

Given $\mathcal{S}$, 
we design a gradient-based method to solve the problem in (\ref{appproblem}).
Recall that we let $\bm{w}_{v} = [w_{vi}, i \in \Omega_v]^{\top}$ be the rating vector for the current injected fake user $v$.
We first determine his/her latent factors by solving Eq.~(\ref{orig_matrix_opti_problem}), which can be restated as:
\begin{align}
\label{mali_matrix_fact}
\mathop {\arg \min }\limits_{{\bm{X}},{\bm{Y}},\bm{z}} {\sum\limits_{(u,i) \in \mathcal{E ^ \prime }} {\left( {{r_{ui}} - \bm{x}_u^\top{\bm{y}_i}} \right)} ^2} 
+ {\sum\limits_{i \in \mathcal{I}} {\left( {{w_{vi}} - {\bm{z}^\top}{\bm{y}_i}} \right)} ^2} \notag \\  
+ \lambda \left( {\sum\nolimits_u {{\lVert \bm{x}_u \rVert}_2^2 } } +  {\sum\nolimits_i {{\lVert \bm{y}_i \rVert}_2^2 } } +  {{\lVert \bm{z} \rVert}_2^2 } \right),
\end{align}
where $\bm{z} \in {\mathbb{R}^d}$ is the latent factor vector for fake user $v$, and $\mathcal{E ^ \prime }$ is the current rating set (rating set $\mathcal{E}$ without attack plus injected ratings of fake users added before user $v$).

Toward this end, note that a subgradient of loss $\mathcal{L}_{\mathcal{S}}(\bm{w}_v)$ in (\ref{appproblem}) can be computed as: 
\begin{align}
\label{rwrOPtMIN}
& G({\bm{w}_v}) = \sum_{u \in \mathcal{S}} {\sum_{i \in \Gamma_u} \nabla_{\bm{w}_v} g(\hat{r}_{ui} - \hat{r}_{ut}) } + \eta {\partial \left\| {{\bm{w}_v}} \right\|_1} \notag \\
&= \sum_{u \in \mathcal{S}} {\sum_{i \in \Gamma_u} {\frac{{\partial g\left({\delta _{u,it}}\right)}}{{\partial {\delta _{u,it}}}}} }\left(\nabla_{\bm{w}_v} \hat{r}_{ui} - \nabla_{\w_v} \hat{r}_{ut} \right) + \eta {\partial \left\| {{\bm{w}_v}} \right\|_1},
\end{align}
where $\delta _{u,it} = \hat{r}_{ui} - \hat{r}_{ut}$ and $\frac{{\partial g\left( \delta _{u,it}\right)}}{{\partial \delta _{u,it}}} = \frac{{g\left( {\delta _{u,it}} \right)\left( {1 - g\left( {\delta _{u,it}} \right)} \right)}}{b}$. 
The subgradient $\partial \left\| {{\bm{w}_v}} \right\|_1$ can be computed as $\frac{\partial }{{\partial {w_{vi}}}}{\left\| {{\bm{w}_v}} \right\|_1} = \frac{{{w_{vi}}}}{{\left| {{w_{vi}}} \right|}}$.
To compute $\nabla_{\w_v} \hat{r}_{ui}$, noting that $\hat{r}_{ui} = \bm{x}_u^\top{\bm{y}_i}$, then the gradient $\frac{{\partial {\hat{r}_{ui}}}}{{\partial {\bm{w}_v}}}$ can be computed as:
\begin{align}
\frac{{\partial {\hat{r}_{ui}}}}{{\partial {\bm{w}_v}}} = \bm{J}_{\w_v}(\x_u)^{\top} \y_{i}  + \bm{J}_{\w_v}(\y_i)^{\top} \x_u,
\end{align}
where $\bm{J}_{\w_v}(\x_u)$ and $\bm{J}_{\w_v}(\y_i)$ are the Jacobian matrices of $\x_u$ and $\y_i$ taken with respect to $\w_v$, respectively.
Next, we leverage first-order stationary condition to approximately compute $\bm{J}_{\w_v}(\x_u)$ and $\bm{J}_{\w_v}(\y_i)$. 
Note that the optimal solution of problem in (\ref{mali_matrix_fact}) satisfies the following first-order stationary condition:
\begin{align} 
{\lambda}{\bm{x}_u} &= \sum\nolimits_{i \in {\Omega _u}} {({r_{ui}} - \bm{x}_u^\top{\bm{y}_i})} {\bm{y}_i} \label{KKT_x_u},
\\
{\lambda}{\bm{y}_i} &= \sum\nolimits_{u \in {\Omega ^i}} {({r_{ui}} - \bm{x}_u^\top{\bm{y}_i})} {\bm{x}_u} + ({w_{vi}} - {\bm{z}^\top}{\bm{y}_i})\bm{z} \label{KKT_y_i},
\\
{\lambda}\bm{z} &= \sum\nolimits_{i \in \mathcal{I}}{(w_{vi} - {\bm{z}^\top}{\bm{y}_i})} {\bm{y}_i} \label{KKT_z},
\end{align}
where ${\Omega _u}$ is the set of items rated by user $u$ and ${\Omega ^i}$ is the set of users who rate the item $i$. 
Inspired by \cite{xiao2015feature,li2016data}, we assume that the optimality conditions given by (\ref{KKT_x_u})--(\ref{KKT_z}) remain valid under an infinitesimal change of $\bm{w}_v$. 
Thus, setting the derivatives of (\ref{KKT_x_u})--(\ref{KKT_z}) with respect to $\bm{w}_v$ to zero and with some algebraic computations, we can derive that: 
\begin{align} 
\frac{{\partial {\bm{x}_u}}}{{\partial {w_{vi}}}} &= \mathbf{0},  \label{x_to_wvi} \\
\frac{{\partial {\bm{y}_i}}}{{\partial w_{vi}}} &= {\left( {{\lambda}{\bm{I}} + \sum\nolimits_{u \in {\Omega ^i}}{{\bm{x}_u}{\bm{x}_u^\top} + \bm{z}\bm{z}^\top} } \right)^{ - 1}}\bm{z},
\label{y_to_wvi}
\end{align}
where $\bm{I}$ is the identity matrix and (\ref{y_to_wvi}) follows from $({\bm{x}_u^\top}\bm{y}_i)\bm{x}_u = (\bm{x}_u{\bm{x}_u^\top})\bm{y}_i$. 
Lastly, computing 
(\ref{x_to_wvi}) and 
(\ref{y_to_wvi}) for all $i\in\Gamma_u$ yields $\bm{J}_{\w_v}(\x_u)$ and $\bm{J}_{\w_v}(\y_i)$.
Note that $\nabla_{\w_v} \hat{r}_{ut}$ can be computed in exactly the same procedure. 
Finally, after obtaining $G({\bm{w}_v})$, we can use the projected subgradient method \cite{Bazaraa_Jarvis_Sherali_90:LP} to solve $\bm{w}_v$ for fake user $v$. 
With $\w_v$, we select the top $n$ items with largest values of $w_{vi}$ as the filler items.
However, the values of $\bm{w}_v$ obtained from solving  (\ref{appproblem}) may not mimic the rating behaviors of normal users. 
To make our $\mathcal{S}$-attack more ``stealthy,'' 
 we will show how to generate rating scores to disguise fake user $v$.
We first set $r_{vt} = r_{max}$ to promote the target item $t$.
Then, we generate rating scores for the filler items by rating each filler item with a normal distribution around the mean rating for this item by legitimate users,  where $\mathcal{N}(\mu_{i},\,\sigma_{i}^{2})$ is the normal distribution with mean $\mu_{i}$ and variance $\sigma_{i}^{2}$ of item $i$. 
Our $\mathcal{S}$-attack algorithm is summarized in Algorithm~\ref{attack_algo}.

\begin{algorithm}[t!]
	\caption{\textit{Our $\mathcal{S}$-Attack.}}\label{attack_algo}
	\begin{algorithmic}[1]
		\renewcommand{\algorithmicrequire}{\textbf{Input:}}
		\renewcommand{\algorithmicensure}{\textbf{Output:}}
		\REQUIRE  Rating matrix $\bm{R}$, target item $t$, parameters $m, n, d, \eta, \lambda, \Delta, b$.
		\ENSURE  Fake user set $\mathcal{M}$.
		\STATE Find influential user set $\mathcal{S}$ according to Algorithm~\ref{Find_Influential_user_Set} for item $t$.
		\STATE Let $\mathcal{M} = \emptyset$. 
		\FOR {$v=1,\cdots,m$} 
			\STATE Solve the optimization problem defined in Eq. (\ref{appproblem}) to get $\bm{w}_{v}$.
			\STATE Select $n$ items with the largest values of $w_{vi}$ as filler items.
		\STATE Set $r_{vt} =  r_{max}$.
		\STATE Let $\mu_{i}$ and $\sigma_{i}^{2}$ be item $i$'s mean and variance of the scores rated by all normal users. Let $r_{vi} \sim \mathcal{N}(\mu_{i},\,\sigma_{i}^{2})$ be the random rating for each filler item $i$ given by fake user $v$. 
		\label{normaldistribution}
		\STATE Let $R\leftarrow R \cup \{\bm{r}_v\}$ and $\mathcal{M} \leftarrow \mathcal{M} \cup \{ v \}$.
		\ENDFOR
		\RETURN $\left\{ {{\bm{r}_v}} \right\}_{v = 1}^m$ and $\mathcal{M}$. 
	\end{algorithmic} 
\end{algorithm}

%% file: experiments.tex

\section{Experiments} \label{sec:exp}

\subsection{Experimental Setup}

\subsubsection{Datasets}
We evaluate our attack on two real-world datasets. 
The first dataset is \textbf{Amazon Digital Music (Music)}~\cite{amazonURL}. This dataset consists of 88,639 ratings on 15,442 music by 8,844 users.
The second dataset is \textbf{Yelp}~\cite{yelpURL}, which contains 504,713 ratings of 11,534 users on 25,229 items.

\subsubsection{$\mathcal{S}$-Attack Variants} 
With different ways of choosing the influential user set $\mathcal{S}$,
we compare three variants of our $\mathcal{S}$-attack.

\myparatight{$\mathcal{U}$-Top-$N$ attack ($\mathcal{U}$-TNA)} This variant uses all normal users as the influential user set $\mathcal{S}$, i.e., $\mathcal{S} = \mathcal{U}$, then solve  Problem  \eqref{appproblem}.

\myparatight{$\mathcal{S}$-Top-$N$ attack+Random ($\mathcal{S}$-TNA-Rand)} This variant randomly selects $\Delta$ users as the influential user set $\mathcal{S}$, then solve  Problem  \eqref{appproblem}.

\myparatight{$\mathcal{S}$-Top-$N$ attack+Influence ($\mathcal{S}$-TNA-Inf)} This variant finds the influential user set $\mathcal{S}$ by Algorithm~\ref{Find_Influential_user_Set}, then solve  Problem  \eqref{appproblem}.

\subsubsection{Baseline Attacks} 

We compare our $\mathcal{S}$-attack variants with the following baseline attacks.

\myparatight{Projected gradient ascent attack (PGA) \cite{li2016data}} PGA attack aims to assign high rating scores to the target items and generates filler items randomly for the fake users to rate. 

\myparatight{Stochastic gradient Langevin dynamics attack (SGLD) \cite{li2016data}} T-his attack also aims to assign high rating scores to the target items, but it mimics the rating behavior of normal users. 
Each fake user will select $n$ items with the largest absolute ratings as filler items. 

\subsubsection{Parameter Setting} Unless otherwise stated, we use the following default parameter setting: $d=64$, $\Delta=400$, $\eta = 0.01$, $b = 0.01$, and $N = 10$. 
Moreover, we set the attack size  to be 3\% (i.e., the number of fake users is 3\% of the number of normal users) 
and the number of filler items is set to $n = 20$. 
We randomly select 10 items as our target items and the hit ratio (HR@$N$) is averaged over the 10 target items, where HR@$N$ of a target item is the fraction of normal users whose top-$N$ recommendation lists contain the target item. 
Note that our  $\mathcal{S}$-attack  is $\mathcal{S}$-TNA-Inf attack.

\begin{table}[!t]
	\centering
	\addtolength{\tabcolsep}{-1.95pt}
	\setlength{\doublerulesep}{2\arrayrulewidth}
	\setlength\extrarowheight{-0.5pt}
	\caption{HR@10 for different attacks.}
	\label{Data_poisoning_attacks_Results}
	\small
	\begin{tabular}{c|c|ccccc}
		\hline
		\multirow{2}[4]{*}{Dataset} & \multirow{2}[4]{*}{Attack} & \multicolumn{5}{c}{Attack size} \bigstrut \\
		\cline{3-7}          &       & 0.3\%     & 0.5\%     & 1\%    & 3\%    & 5\% \bigstrut \\
		\hline
		\hline
		\multirow{6}[2]{*}{Music} 
		& None  & 0.0017 & 0.0017 & 0.0017 &0.0017 & 0.0017 \\
		& PGA \cite{li2016data}  & 0.0107 & 0.0945 & 0.1803 & 0.3681 & 0.5702 \\
		& SGLD \cite{li2016data} & 0.0138 & 0.1021 & 0.1985 & 0.3587 & 0.5731 \\
		& $\mathcal{U}$-TNA & 0.0498 & 0.1355 & 0.2492 & 0.4015 & 0.5832 \\
		& $\mathcal{S}$-TNA-Rand & 0.0141 & 0.0942 & 0.2054 & 0.3511 & 0.5653 \\
		& $\mathcal{S}$-TNA-Inf  & \textbf{0.0543} & \textbf{0.1521} & \textbf{0.2567} & \textbf{0.4172} & \textbf{0.6021} \\
		\hline
		\multirow{6}[2]{*}{Yelp} 
		& None  & 0.0015 & 0.0015 & 0.0015 & 0.0015 & 0.0015 \\
		& PGA \cite{li2016data}  & 0.0224 & 0.1623 & 0.4162 & 0.4924 & 0.6442 \\
		& SGLD \cite{li2016data} & 0.0261 & 0.1757 & 0.4101 & 0.5131 & 0.6431 \\
		& $\mathcal{U}$-TNA & 0.0619 & \textbf{0.2304} & 0.4323 & 0.5316 & 0.6806 \\
		& $\mathcal{S}$-TNA-Rand & 0.0258 & 0.1647 & 0.4173 & 0.4923 & 0.6532 \\
		& $\mathcal{S}$-TNA-Inf  & \textbf{0.0643} & 0.2262 & \textbf{0.4415} & \textbf{0.5429} & \textbf{0.6813} \\
		\hline
		\hline
	\end{tabular}%
\end{table}%

\subsection{Full-Knowledge Attack}
In this section, we consider the worst-case attack scenario, where the attacker has full knowledge of the recommender system, e.g., the type of the target recommender system (matrix-factorization-based), all rating data, and the parameters of the recommender system (e.g., the dimension $d$
and the tradeoff parameter $\lambda$ in use).

Table~\ref{Data_poisoning_attacks_Results} summaries the results of different attacks. 
``None'' means the hit ratios without any attacks. 
First,  we observe that the variants of our $\mathcal{S}$-attack can effectively promote the target items using only a small number of fake users. 
For instance, in the Yelp dataset, when injecting only 0.5\% of fake users, $\mathcal{S}$-TNA-Inf attack improves the hit ratio by 150 times for a random target item compared to that of the non-attack setting. 
Second,  the variants of our $\mathcal{S}$-attack outperform the baseline attacks in most cases. 
This is because the baseline attacks aim to manipulate all the missing entries of the rating matrix, while our attack aims to manipulate the top-$N$ recommendation lists. 
Third, it is somewhat surprising to see that the $\mathcal{S}$-TNA-Inf attack outperforms the $\mathcal{U}$-TNA attack. 
Our observation shows that by dropping the users that are not influential to the recommendation of the target items when optimizing the rating scores for the fake users, we can improve the effectiveness of our attack.

\vspace{-.01in}
\subsection{Partial-Knowledge Attack}

In this section, we consider partial-knowledge attack. 
In particular, we consider the case where the attacker knows the type of the target recommender system (matrix-factorization-based), but the attacker has access to a subset of the ratings for the normal users and does not know the dimension $d$. In particular, we view the user-item rating matrix as a bipartite graph. Given a size of observed data, we construct the subset of ratings by selecting nodes (users and items) with increasing distance from the target item (e.g., one-hop distance to the target item, then two-hop distance and so on) on the bipartite graph until we reach the size of observed data.

\begin{figure}[t]
	\centering
	\begin{subfigure}[b]{0.23\textwidth}
		\includegraphics[width=\textwidth]{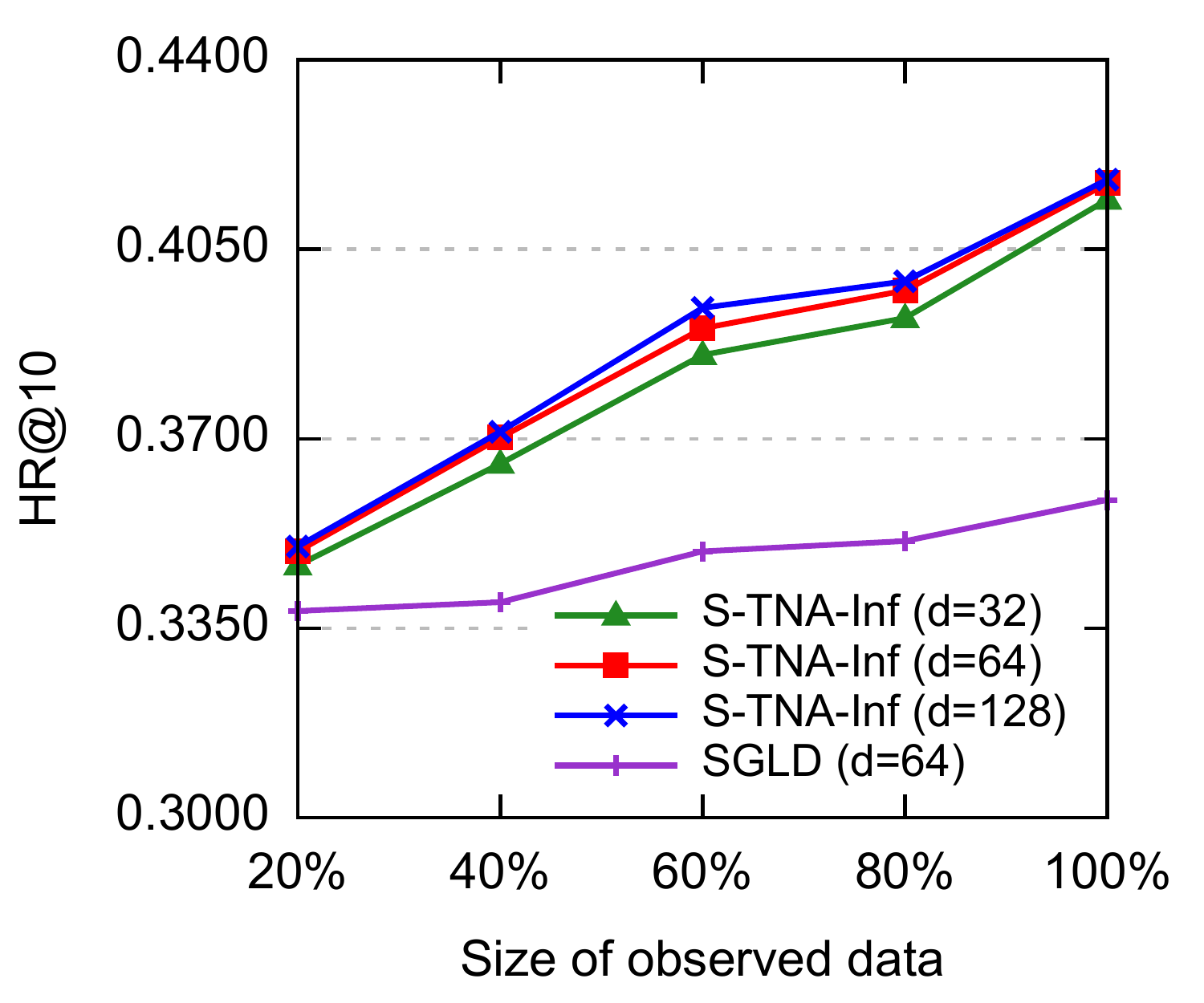}
		\caption{Music}
	\end{subfigure}
	\begin{subfigure}[b]{0.23\textwidth}
		\includegraphics[width=\textwidth]{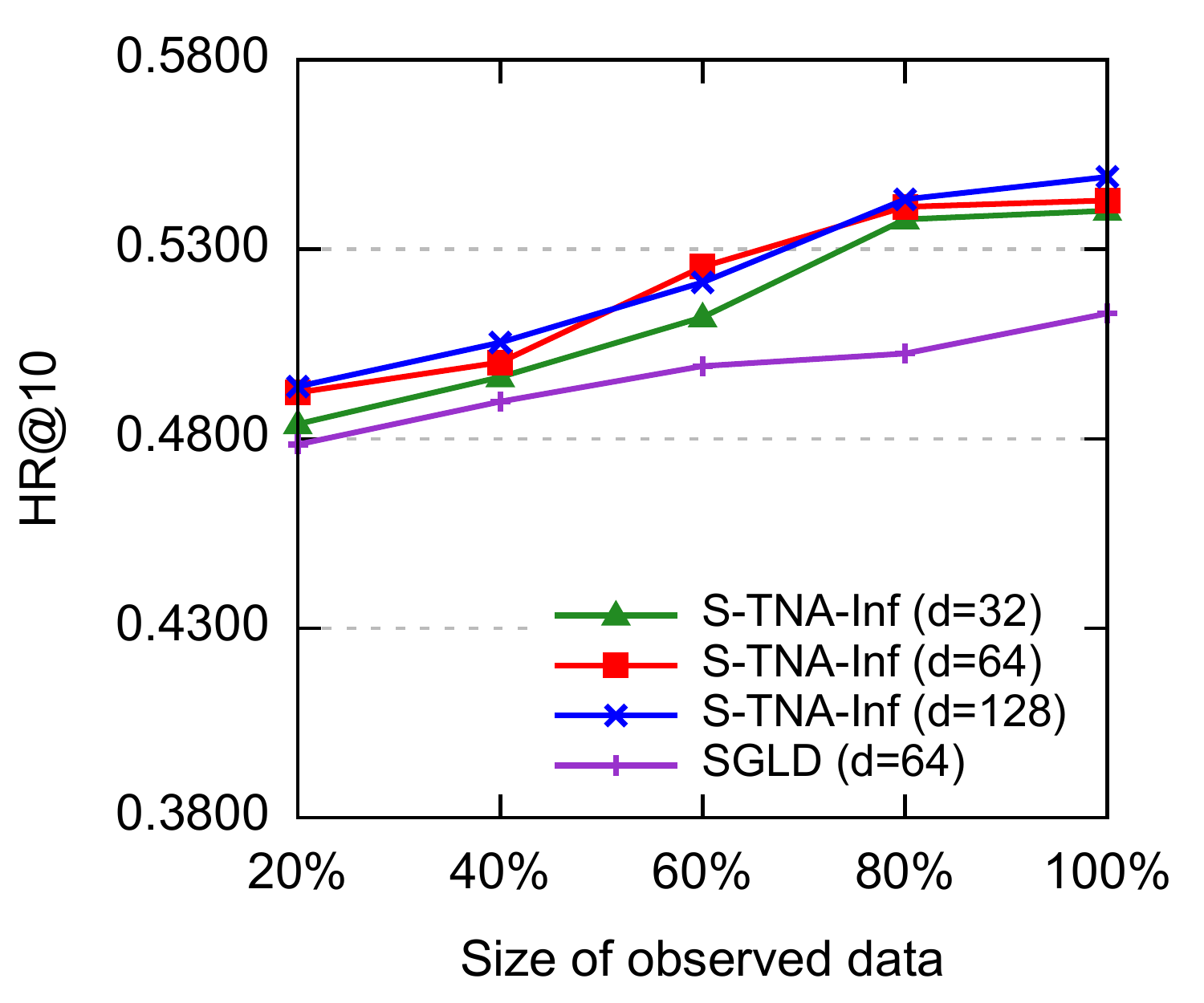}
		\caption{Yelp}
	\end{subfigure}
	\caption{The attacker knows a subset of ratings for the normal users and does not know $d$.} 
	\label{partial_knowledge}
\end{figure}

Figure~\ref{partial_knowledge} shows the attack results when the attacker observes different amounts of normal users ratings and our attack uses different $d$, where the target recommender system uses $d=64$. The attack size is set to be 3\%.
Note that in the partial-knowledge attack, the attacker selects the influential user set and generates fake users based only on the observed data. 
Naturally, we observe that as the attacker has access to more ratings of the normal users, the attack performance improves. 
We find that our attack also outperforms SGLD attack (which performs better than PGA attack) in the partial-knowledge setting. 
Moreover, our attack is still effective even if the attacker does not know $d$. In particular, the curves corresponding to different $d$ are close to each other for our attack in Figure~\ref{partial_knowledge}. 

\begin{figure}[t]
	\centering
	\begin{subfigure}[b]{0.23\textwidth}
		\includegraphics[width=\textwidth]{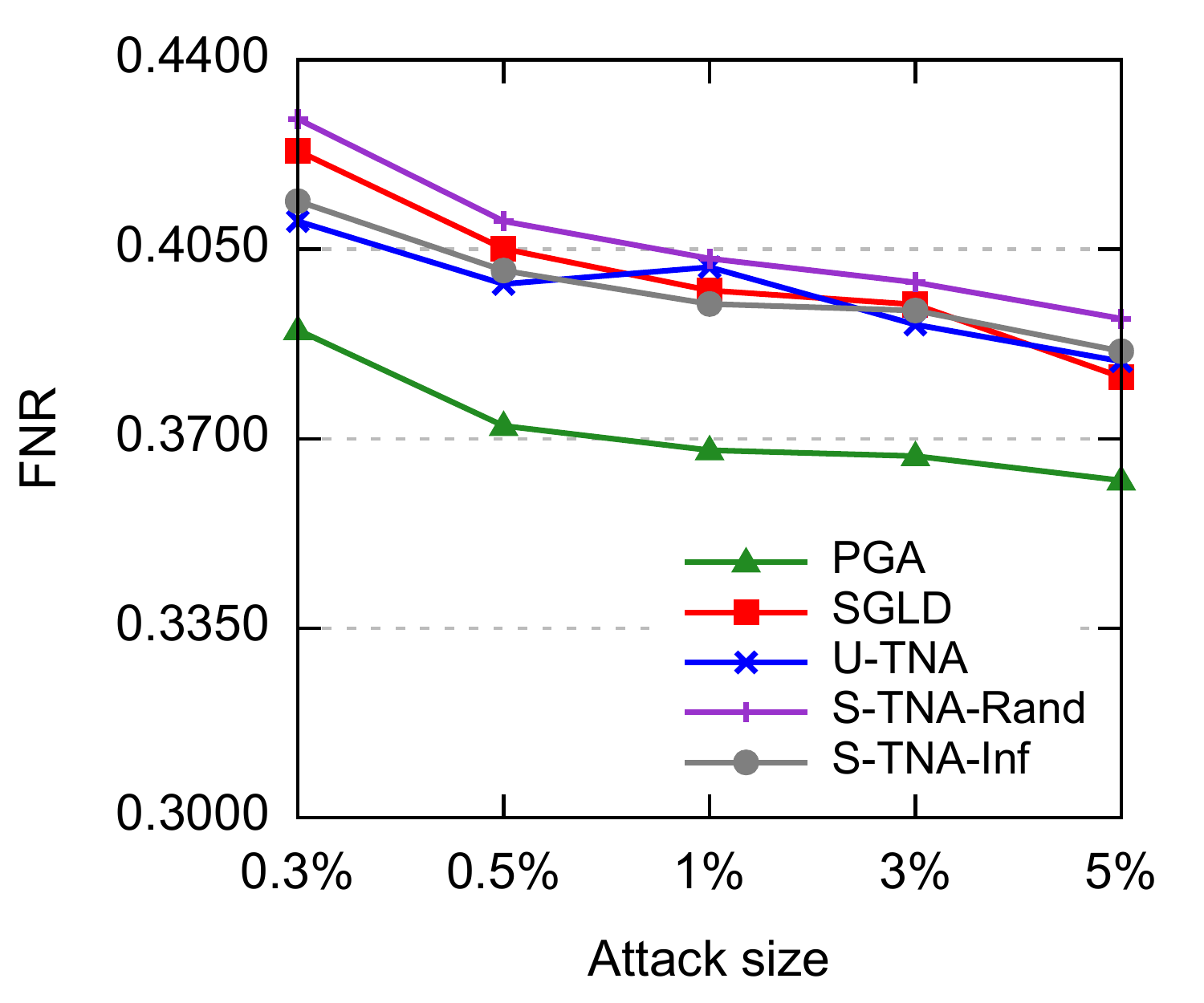}
		\caption{Music}
	\end{subfigure}
	\begin{subfigure}[b]{0.23\textwidth}
		\includegraphics[width=\textwidth]{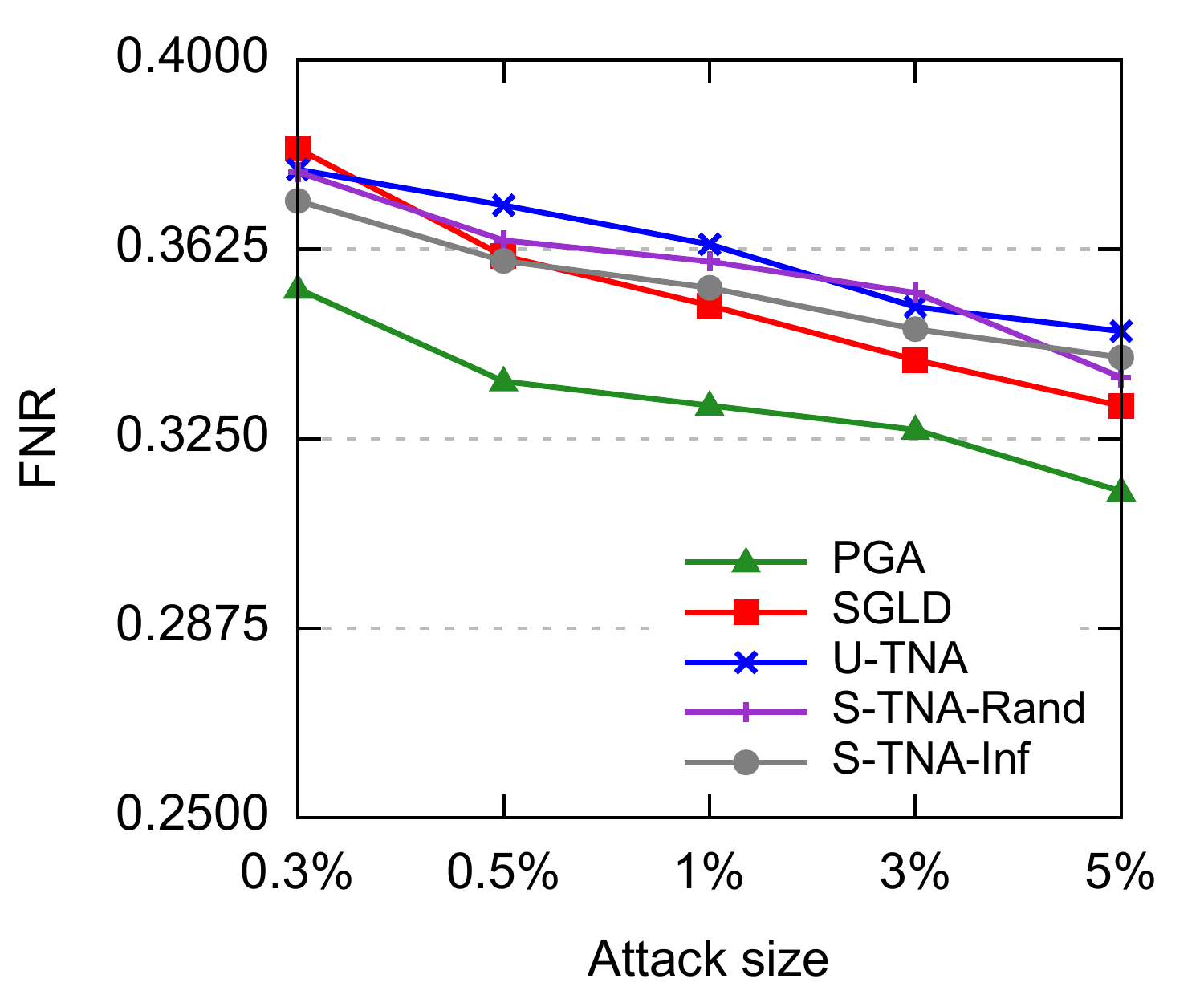}
		\caption{Yelp}
	\end{subfigure}
	\caption{FNR scores for different attacks.} 
	\label{FNR_scores}
\end{figure}

\section{Detecting Fake Users}
\label{sec:attack_fake_recom}
To minimize the impact of potential attacks on recommender systems, a service provider may arm the recommender systems with certain fake-user detection capability. 
In this section, we investigate whether our attack is still effective in attacking the fake-user-aware recommender systems. 
Specifically, we extract six features--namely, RDMA~\cite{chirita2005preventing}, WDMA~\cite{mobasher2007toward}, WDA~\cite{mobasher2007toward}, TMF~\cite{mobasher2007toward}, FMTD~\cite{mobasher2007toward}, and MeanVar~\cite{mobasher2007toward}--for each user from its ratings.
Then, for each attack, we construct a training dataset consisting of 800 fake users generated by the attack and 800 randomly sampled normal users. We use the training dataset to learn a SVM classifier. Note that the classifier may be different for different attacks. 

\myparatight{Fake-user detection results} We deploy the trained SVM classifiers to detect the fake users under different attacks settings. 
Figure~\ref{FNR_scores} reports the fake users detection results of different attacks, where False Negative Rate (FNR) represents the fraction of fake users that are predicted to be normal. 
From Figure~\ref{FNR_scores}, we find that 
PGA attack is most likely to be detected. The reason is that the fake users generated by PGA attack do not rate the filler items according to normal users' behavior, thus the generated fake users are easily detected. 
We also observe that a large fraction of fake users are not detected. 

\myparatight{Attacking fake-user-aware recommender systems} 
We now test the performance of attacks on fake-user-aware recommender systems. 
Suppose that the service provider removes the predicted fake users from the system detected by the trained SVM classifiers. 
We recompute the hit ratio after the service provider excludes the predicted fake users from the systems. 
Note that a large portion of fake users and a small number of normal users will be deleted.
The results are shown in Table~\ref{attack_under_detection}. 
We observe that PGA attack achieves the worst attack performance when the service provider removes the predicted fake users from the systems.
The reason is that the PGA attack is most likely to be detected. 
Comparing Table~\ref{Data_poisoning_attacks_Results} and Table~\ref{attack_under_detection}, we can see that when the target recommender system is equipped with fake-user detectors, our attacks remain effective in promoting the target items and outperform the baseline attacks. 
This is because the detectors miss a large portion of the fake users.

\begin{table}[t!]
	\centering
	\addtolength{\tabcolsep}{-1.95pt}
	\setlength{\doublerulesep}{2\arrayrulewidth}
	\setlength\extrarowheight{-0.5pt}
	\caption{HR@10 for different attacks when attacking the fake-user-aware recommender systems.}
	\label{attack_under_detection}
	\small
	\begin{tabular}{c|c|ccccc}
		\hline
		\multirow{2}[4]{*}{Dataset} & \multirow{2}[4]{*}{Attack} & \multicolumn{5}{c}{Attack size} \bigstrut \\
		\cline{3-7}          &       & 0.3\%     & 0.5\%     & 1\%    & 3\%    & 5\% \bigstrut \\
		\hline
		\hline
		\multirow{6}[2]{*}{Music} 
		& None  & 0.0011 & 0.0011 & 0.0011 & 0.0011 & 0.0011 \\
		& PGA \cite{li2016data}  & 0.0028 & 0.0043 & 0.0311 & 0.2282 & 0.3243 \\
		& SGLD \cite{li2016data} & 0.0064 & 0.0145 & 0.0916 & 0.2631 & 0.3516 \\
		& $\mathcal{U}$-TNA & 0.0127 & 0.0298 & \textbf{0.1282} & 0.2846 & 0.3652 \\
		& $\mathcal{S}$-TNA-Rand & 0.0068 & 0.0139 & 0.0934 & 0.2679 & 0.3531 \\
		& $\mathcal{S}$-TNA-Inf  & \textbf{0.0199} & \textbf{0.0342} & 0.1215 & \textbf{0.2994} & \textbf{0.3704} \\
		\hline
		\multirow{6}[2]{*}{Yelp} 
		& None  & 0.0010 & 0.0010 & 0.0010 & 0.0010 & 0.0010 \\
		& PGA \cite{li2016data}  & 0.0018 & 0.0062 & 0.1143 & 0.3301 & 0.4081 \\
		& SGLD \cite{li2016data} & 0.0097 & 0.0278 & 0.1585 & 0.3674 & 0.4223 \\
		& $\mathcal{U}$-TNA & 0.0231 & 0.0431 & 0.1774 & 0.3951 & 0.4486 \\
		& $\mathcal{S}$-TNA-Rand & 0.0093 & 0.0265 & 0.1612 & 0.3665 & 0.4269 \\
		& $\mathcal{S}$-TNA-Inf  & \textbf{0.0242} & \textbf{0.0474} & \textbf{0.1831} & \textbf{0.3968} & \textbf{0.4501} \\
		\hline
		\hline
	\end{tabular}%
\end{table}%

%% file: discussion.tex

\section{Discussion} \label{sec:Discussion}

We show that our influence function based approach can be extended to enhance data poisoning attacks to graph-based top-$N$ recommender systems. In particular, we select a subset of normal users based on influence function and optimize data poisoning attacks using them. Moreover, we show that an attacker can also use influence function to weight each normal user instead of selecting the most influential ones, which sacrifices computational efficiency but achieves even better attack effectiveness. 

\subsection{Influence Function for Graph-based Recommender Systems}

We investigate whether we can extend our influence function based method to optimize data poisoning attacks to graph-based recommender systems~\cite{fang2018poisoning}. 
Specifically, we aim to find a subset of users who have the largest impact on the target items in graph-based recommender systems. 
It turns out that, when optimizing the attack over these subset of users, we obtain better attack effectiveness. 
Toward this end, we will first show how to find a subset of influential users for the target items in graph-based recommender systems. 
Then, we optimize the attack proposed by \cite{fang2018poisoning} over the subset of influential users.

We consider a graph-based recommender system using  random walks~\cite{pirotte2007random}. 
Specifically, the recommender system models the user-item ratings as a bipartite graph, where a node is a user or an item,  an edge between a user and an item means that the user rated the item, and an edge weight is the corresponding rating score. 
We let $\bm{p}_u$ represent the stationary distribution of a random walk with restart that starts from the user $u$ in the bipartite graph. 
Then, $\bm{p}_u$ can be computed by solving the following equation:
\begin{align}
\label{RWR}
\bm{p}_u = (1 - \alpha ) \cdot \bm{Q} \cdot \bm{p}_u + \alpha  \cdot \bm{e}_u,
\end{align}
where $\bm{e}_u$ is a basis vector whose $u$-th entry is 1 and all other entries are 0,  $\bm{Q}$ is the \emph{transition matrix}, and $\alpha$ is the restart probability. We let $q_{ui}$ denote the value at $(u,i)$-th entry of matrix $\bm{Q}$. Then $q_{ui}$ can be computed as:
\begin{align}
\label{matrixQ}
{q_{ui}} = \begin{cases}
\frac{{{r_{ui}}}}{{\sum\nolimits_{j} {{r_{uj}}} }}, & \mbox{if $(u, i) \in \mathcal{E} $},\\
0, & \mbox{otherwise.}
\end{cases}
\end{align}

The $N$ items that were not rated by user $u$ and that have the largest probabilities in the stationary distribution  $\bm{p}_u$ are recommended to $u$. 
We define the influence of removing edge $(k,j) \in \mathcal{E}$ in the user-item bipartite graph on the target item $t$ when performing a random walk from user $u$ as the change of prediction at $p_{ut}$ upon removing edge $(k,j)$:
\begin{align}
\label{influ_put_deri_Q}
\delta((k,j),p_{ut}) \stackrel{\text{def}} = \frac{{\partial {p_{ut}}}}{{\partial q_{kj}}},
\end{align}
where $q_{kj}$ is the transition matrix entry as defined in (\ref{matrixQ}). According to (\ref{RWR}), $\frac{{\partial {\bm{p}_{u}}}}{{\partial q_{kj}}}$ can be computed as:
\begin{align}
\label{influ_p_deri_Q}
\frac{{\partial {\bm{p}_u}}}{{\partial q_{kj}}} = (1 - \alpha )\frac{{\partial \bm{Q}}}{{\partial q_{kj}}}{\bm{p}_u} + (1 - \alpha )\bm{Q}\frac{{\partial {\bm{p}_u}}}{{\partial q_{kj}}}.
\end{align}
\noindent After rearranging terms in (\ref{influ_p_deri_Q}), we have:
\begin{align}
\frac{{\partial {\bm{p}_u}}}{{\partial q_{kj}}} = (1 - \alpha ){\left( {\bm{I} - (1 - \alpha )\bm{Q}} \right)^{ - 1}}\frac{{\partial \bm{Q}}}{{\partial q_{kj}}}{\bm{p}_u},
\end{align}
\noindent where $\bm{I}$ is the identity matrix, $\frac{{\partial \bm{Q}}}{{\partial q_{kj}}}$ is a single-nonzero-entry matrix with
its $(k,j)$-th entry being 1 and 0 elsewhere. 
By letting $\bm{M} \triangleq  {\left( {\bm{I} - (1 - \alpha )\bm{Q}} \right)^{ - 1}}$, we have the following:
\begin{align}
\frac{{\partial {\bm{p}_u}}}{{\partial q_{kj}}} = (1 - \alpha ){p_{uj}}\bm{M}(:,k),
\end{align}
where $\bm{M}(:,k)$ is the $k$-th column of matrix $\bm{M}$. 
Then, the influence of removing edge $(k,j)$ on the prediction at the target item $t$ when performing a random walk from user $u$ can be calculated as:
\begin{align}
\delta((k,j),p_{ut}) =\frac{{\partial {p_{ut}}}}{{\partial q_{kj}}} = (1 - \alpha ){p_{uj}}M(t,k).
\end{align}
Therefore, the influence of removing edge $(k,j)$ on the prediction at the target item $t$ can be computed as:
\begin{align}
\label{edge_influ_target_graph}
\varphi((k,j),t)  & \stackrel{\text{def}}= \sum\limits_{u \in \mathcal{U}} {\delta((k,j),p_{ut})}.
\end{align}
We could approximate matrix $\bm{M}$ by using Taylor expansion 
$\bm{M} = {\left( {\bm{I} - (1 - \alpha )\bm{Q}} \right)^{ - 1}} \approx \bm{I} + \sum\nolimits_{i = 1}^T  {{{(1 - \alpha )}^i}} {\bm{Q}^i}$.
For example, we can choose $T=1$ if we use first order Taylor approximation.

After obtaining $\varphi((k,j),t)$, we can compute the influence of user $k$ at the target item $t$, namely $\pi(k,t)$, based on (\ref{influence_user_target}).
Then, we apply Algorithm \ref{Find_Influential_user_Set} to approximately find an influential user set $\mathcal{S}$. 
With the influential user set $\mathcal{S}$, we can optimize the attack  proposed by \cite{fang2018poisoning} over the most influential user set and compare with the attack proposed by \cite{fang2018poisoning}, which uses all normal users. 
The poisoning attack results of graph-based recommender systems are shown in Table~\ref{compare_fang}, where the experimental settings are the same as those in \cite{fang2018poisoning}.
Here, ``None'' in Table~\ref{compare_fang} means the hit ratios without attacks computed in graph-based recommender systems; and 
 ``$\mathcal{S}$-Graph'' means optimizing the attack  proposed by \cite{fang2018poisoning} over the most influential users in $\mathcal{S}$, where we select 400 influential users. 
From Table~\ref{compare_fang}, we  observe that the optimized attacks based on influence function outperform existing ones \cite{fang2018poisoning}. 

\begin{table}[!t]
	\centering
	\addtolength{\tabcolsep}{-1.95pt}
	\setlength{\doublerulesep}{2\arrayrulewidth}
	\setlength\extrarowheight{-0.5pt}
	\caption{HR@10 for attacks to graph-based recommender systems.}
	\label{compare_fang}
	\small
	\begin{tabular}{c|c|ccccc}
		\hline
		\multirow{2}[4]{*}{Dataset} & \multirow{2}[4]{*}{Attack} & \multicolumn{5}{c}{Attack size} \bigstrut  \\
		\cline{3-7}          &       & 0.3\%     & 0.5\%     & 1\%    & 3\%    & 5\% \bigstrut \\
		\hline
		\hline
		\multirow{3}[2]{*}{Music} 
		& None  & 0.0021 & 0.0021 & 0.0021 &  0.0021 & 0.0021 \\
		& Fang et. al \cite{fang2018poisoning} & 0.0252   & 0.1021 & 0.2067 &  0.2949 & 0.5224 \\
		& $\mathcal{S}$-Graph & 0.0245     & \textbf{0.1046} & \textbf{0.2125} &  \textbf{0.3067} & \textbf{0.5368} \\
		\hline
		\multirow{3}[2]{*}{Yelp} 
		& None  & 0.0023     & 0.0023 & 0.0023 &  0.0023 & 0.0023 \\
		& Fang et. al \cite{fang2018poisoning} & 0.0256  & 0.1359 & 0.2663 &  \textbf{0.4024} & 0.5704 \\
		& $\mathcal{S}$-Graph & \textbf{0.0342}  & \textbf{0.1514} & \textbf{0.2701} &  0.4011 & \textbf{0.5723} \\
		\hline
		\hline
	\end{tabular}%
\end{table}%

\subsection{Weighting Normal Users}

In this section, we show that our approach can be extended to a more general framework: we can weight the normal users instead of dropping some of them using the influence function. More specifically, we optimize our attack  over all normal users, and different normal users are assigned different weights in the objective function based on their importance with respect to the target item. 
Intuitively, the important users should receive more penalty if the target item does not appear in those users' recommendation lists.
Toward this end, we let $\bm{\mathcal{H}} = [\mathcal{H}_u, u \in \mathcal{U}]^{\top}$ be the weight vector for all normal users, then we can modify the loss function defined in (\ref{appproblem}) as:
\begin{align}
\label{appproblem_reweight}
\begin{split}
\min_{\w_v,\forall v} \mathcal{L}_{\mathcal{U}}(\bm{w}_v) &= \sum_{u \in \mathcal{U}} {\sum_{i \in \Gamma_u} g(\hat{r}_{ui} - \hat{r}_{ut})} \cdot \mathcal{H}_u +  \eta {\left\| \bm{w}_v \right\|_1} \\
\text{s.t. } & \sum_{u \in \mathcal{U}} \mathcal{H}_u = 1, \\
& w_{vi} \in [0, r_{max}],
\end{split}
\end{align}
where $\mathcal{H}_u$ is the weight for normal user $u$ and satisfies $\mathcal{H}_u \ge 0$. We can again leverage the influence function technique to compute the weight vector $\bm{\mathcal{H}}$. 
For a normal user $k$, the weight can be computed in a normalized fashion as follows:
\begin{align}
\label{user_weight_vector}
\mathcal{H}_k = \frac{\pi(k,t)}{\sum_{u \in \mathcal{U}} \pi(u,t)},
\end{align}
where $\pi(k,t)$ is the influence of user $k$ at the target item $t$, and can be computed according to (\ref{influence_user_target}). Note that here we compute $\pi(k,t)$ for each user $k$ at one time.

After obtaining the weight vector $\bm{\mathcal{H}}$, we can compute the derivative of function defined in (\ref{appproblem_reweight}) in a similar way. Table~\ref{Reweighting_Training_Instances} illustrates the attack results on matrix-factorization-based recommender systems when we weight normal users, where the experimental settings are the same as those in Table~\ref{Data_poisoning_attacks_Results}. 
Here, ``Weighting" means that we weight each normal user and optimize the attack of (\ref{appproblem_reweight}) over the weighted normal users, and the weight of each normal user is computed based on (\ref{user_weight_vector}). 
Comparing Tables~\ref{Data_poisoning_attacks_Results} and \ref{Reweighting_Training_Instances}, we can see that the performance is improved when we consider the weights of different normal users with respect to the target items. Our results show that, when an attacker has enough computational resource, the attacker can further improve attack effectiveness using influence function to weight normal users instead of dropping some of them.

\begin{table}[!t]
	\centering
	\addtolength{\tabcolsep}{-1.95pt}
	\setlength{\doublerulesep}{2\arrayrulewidth}
	\setlength\extrarowheight{-0.5pt}
	\caption{HR@10 for weighting-based attacks to matrix-factorization-based recommender systems.}
	\label{Reweighting_Training_Instances}
	\small
	\begin{tabular}{c|c|ccccc}
		\hline
		\multirow{2}[4]{*}{Dataset} & \multirow{2}[4]{*}{Attack} & \multicolumn{5}{c}{Attack size} \bigstrut \\
		\cline{3-7}          &       & 0.3\%     & 0.5\%     & 1\%    & 3\%    & 5\% \bigstrut \\
		\hline
		\hline
		Music & Weighting  & \textbf{0.0652} & \textbf{0.1543} & 0.2436 &  \textbf{0.4285} & \textbf{0.6087} \\
		\hline
		Yelp  & Weighting & \textbf{0.0698}  & \textbf{0.2312} & \textbf{0.4498} &  \textbf{0.5501} & \textbf{0.6924} \\
		\hline
		\hline
	\end{tabular}%
	\label{tab:addlabel}%
\end{table}%

%% file: conclusion.tex

\section{Conclusion} \label{sec:conclusion}

In this paper, we proposed the first data poisoning attack to matrix-factorization-based top-$N$ recommender systems. 
Our key idea is that, instead of optimizing the ratings of a fake user using all normal users, we use a subset of influential users. Moreover, we proposed an efficient influence function based method to determine the influential user set for a specific target item.
We also performed extensive experimental studies to demonstrate the efficacy of our proposed attacks.
Our results showed that our proposed attacks outperform existing ones.
%
\section*{Acknowledgements}
This work has been supported in part by NSF grants ECCS-1818791, CCF-1758736, CNS-1758757, CNS-1937786; ONR grant N00014-17-1-2417, and AFRL grant FA8750-18-1-0107.